\documentclass[a4paper, 11.5pt]{article}
\pdfoutput=1
\usepackage{amsthm,amsmath,amsfonts,amssymb}
\usepackage[authoryear]{natbib}
\usepackage[colorlinks,citecolor=blue,urlcolor=blue]{hyperref}
\usepackage{graphicx}

\usepackage[utf8]{inputenc} 
\usepackage[T1]{fontenc}    
\usepackage[subscriptcorrection]{newtxmath}
\usepackage[plain,noend]{algorithm2e}
\usepackage{geometry, setspace} 

\usepackage{url, booktabs, nicefrac, caption, subcaption, cprotect, enumitem, bbm, tikz, adjustbox, wrapfig, xr, changepage, datetime, comment, physics}
\let\dv\undefined

\usetikzlibrary{arrows, automata, calc, positioning}
\usepackage[multiple]{footmisc}

\usepackage{float, mathtools, bm, thmtools, multirow, etoolbox, xpatch, anyfontsize}

\theoremstyle{plain}

\newtheorem{claim}{Claim}
\newtheorem{theorem}{Theorem}
\newtheorem{lemma}{Lemma}
\newtheorem{proposition}{Proposition}
\newtheorem{corollary}{Corollary}

\newtheorem{assumption}{Assumption}

\theoremstyle{remark}
\newtheorem{definition}{Definition}

\newtheorem{remark}{Remark}

\newcommand{\indep}{\perp \!\!\! \perp}

\newcommand{\iidsim}{\overset{\text{iid}}{\sim}}
\newcommand{\E}[1]{\mathbb{E}\left\{#1\right\}}

\newcommand{\dv}{\ \mathrm{d}}

\newcommand{\Cov}[1]{\mathrm{Cov}\left(#1\right)}
\newcommand{\Var}[1]{\mathrm{Var}\left(#1\right)}

\newcommand{\Norm}[1]{\mathcal{N}\left(#1\right)}

\newcommand{\Bern}[1]{\mathrm{Bern}\left(#1\right)}

\newcommand{\Gam}[1]{\mathrm{Gamma}\left(#1\right)}

\newcommand{\LN}[1]{\mathcal{LN}\left(#1\right)}

\let\originalleft\left
\let\originalright\right
\renewcommand{\left}{\mathopen{}\mathclose\bgroup\originalleft}
\renewcommand{\right}{\aftergroup\egroup\originalright}

\title{Constructing an Instrument as a Function of Covariates}
\author{Moses Stewart \\ 
        {Harvard College\footnote{Advised by Rahul Singh. Thanks to Jesse Shapiro for his guidance throughout the research process and Isaiah Andrews for his comments.}}
        }
        
\date{\monthname\ \number\year}

\begin{document}
\maketitle

\begin{abstract}
\noindent 
Researchers often use instrumental variables (IV) models to investigate the causal relationship between an endogenous variable and an outcome while controlling for covariates. When an exogenous variable is unavailable to serve as the instrument for an endogenous treatment, a recurring empirical practice is to construct one from a nonlinear transformation of the covariates. We investigate how reliable these estimates are under mild forms of misspecification. Our main result shows that for instruments constructed from covariates, the IV estimand can be arbitrarily biased under mild forms of misspecification, even when imposing constant linear treatment effects. We perform a semi-synthetic exercise by calibrating data to alternative models proposed in the literature and estimating the average treatment effect. Our results show that IV specifications that use instruments constructed from covariates are non-robust to nonlinearity in the true structural function.
\end{abstract}


\begin{spacing}{1.45}
\section{Introduction}\label{sec:intro}

Instrumental variable (IV) strategies are widely used to investigate the causal effect of an endogenous variable on an outcome of interest. Since the work of \cite{imbens_identification_1994}, these estimators are often interpreted as estimating a local average treatment effect (LATE) or, under homogeneity conditions, the population treatment effect of the endogenous variable of interest. 

However, estimating this causal effect requires finding valid instruments that satisfy the exclusion restriction, which can be challenging in many applications. One response to this challenge is to construct an instrumental variable that is a nonlinear function of covariates.\footnote{Here we use covariate to mean ``included exogenous variable," and reserve ``excluded variable'' to refer to excluded exogenous variables.}  The leading case is to use the ``product instrument," the interaction of two covariates. A sample of recent articles published in the \textit{American Economic Review} that use this approach are \cite{munshi_networks_2016}, \cite{bettinger_virtual_2017}, \cite{aladangady_housing_2017}, and \cite{rogall_mobilizing_2021}. Other examples include \cite{shao_estimating_2014}, \cite{bharadwaj_fertility_2015}, \cite{helmers_board_2017}, \cite{crawford_impact_2020}, \cite{cage_social_2022}, \cite{liu_housing_2023}, and \cite{ghose_leveraging_2024}.

In this paper, we study the causal interpretability of these IV estimators that use a nonlinear function of covariates as the instrumental variable, under misspecification. A researcher may be interested in the effect of a variable $D$ on some outcome $Y$, where $D$ may be endogenous to unobserved factors affecting $Y$. The researcher’s model specifies $Y$ as a linear function of $D$ and some covariates $X$, with the causal relationships governed by a parameter $\tilde{\theta}$. They do not have access to an excluded variable $Z$ that is correlated with the endogenous variable $D$ but not the unobserved factors affecting the relationship between $D$ and $Y$ (conditional upon $X$). Instead, a recurring practice is for the researcher to construct an instrumental variable $\tilde{f}(X) = X_{1} \cdot X_{2}$ that is the interaction of covariates, which they argue satisfy the IV assumptions. The researcher then uses IV estimators to evaluate $\tilde{\theta}$.

For example, \cite{aladangady_housing_2017} uses a linear model to examine the effect of log house price fluctuations $D$ on log consumer spending $Y$ in the United States while trying to control for potential unobserved confounding caused by preferences or expectations that change over time. They would like to use a measurement of local land availability or zoning regulations $Z$ that is uncorrelated with preference shifts but varies geographically. However, they do not have access to detailed zoning data across regions. Therefore, they use the interaction between long-term real interest rates $X_{1}$ and housing supply elasticity measures $X_{2}$ as an instrument for house price fluctuations and estimate the two-stage least squares (2SLS) model,
\begin{equation*}\begin{aligned}
    Y = \tilde{\pi}_{0} + \tilde{\theta} D + \tilde{\pi}_{1} X_{1} + \tilde{\pi}_{2} X_{2} + \tilde{\epsilon}_{i}, \qquad\qquad
    D = \tilde{\tau}_{0} +  \tilde{\delta} \tilde{f}(X) + \tilde{\tau}_{1} X_{1} + \tilde{\tau}_{2} X_{2} + \tilde{\nu}_{i},
\end{aligned}\end{equation*}


\noindent for $\tilde{f}(X) = X_{1} \cdot X_{2}$. By including long-term real interest rates $X_{1}$ and housing supply elasticity measures $X_{2}$ in the specification for consumer spending $Y$, they hope to control for the \textit{direct effect} of interest rates and housing elasticity on spending trends that may be related to housing prices. This relaxes the exogeneity assumption required on the instrument $\tilde{f}(X)$ because even if long-term interest rates or housing elasticity are correlated with unobserved determinants of spending, as long as their interaction is uncorrelated with these unobserved determinants, the model assumptions required for IV are still satisfied.

We show that for IV estimators that use a function of covariates as an instrument, even while imposing homogeneous treatment effects, the researcher's model is uninterpretable as a causal effect under mild forms of misspecification. As the main contribution we show that for any level of bias and for any IV estimand which relies on an instrument that is a transformation of covariates $f(X, Z) = \tilde{f}(X)$, there is a continuous data generating process (DGP) with constant, linear treatment effects under which we can attain that level of bias. 

Quantitatively, we investigate the sensitivity of linear IV estimates that use a function of covariates $\tilde{f}(X)$ as an instrument to a reasonable set of nonlinear DGPs. For a collection of influential papers that use linear IV with product instruments, we perform a semi-synthetic exercise by estimating the treatment effect of the endogenous variable $D$ under nonlinear structural functions proposed in the literature. For three out of four pairs of papers we examine, we find that a 95\% confidence interval for the treatment effect $\hat{\theta}$ rejects the ground truth in our simulations. These findings emphasize the sensitivity of IV estimands that use product instruments to the true relationship between the outcome $Y$ and covariates $X$ in the true structural function. This paper aims to guide the causal interpretation of these estimands. 

A large literature following \cite{imbens_identification_1994} and \cite{angrist_identification_1996}
studies the interpretation of IV estimators under misspecification. Our work is closest to \cite{blandhol_when_2022}, which investigates the LATE interpretation of IV estimators that include exogenous covariates $X$. Unlike \cite{blandhol_when_2022}, we allow for continuous outcomes and continuous treatment effects. \cite{blandhol_when_2022} shows that IV estimators with instruments that do not satisfy a condition called ``saturated covariates'' cannot be interpreted as LATE. As we show in section~\ref{sec:model}, the estimators we study are undefined under saturated covariates since there is no variation in the instrument $\tilde{f}(X)$. Appendix~\ref{sec:blandhol} expands on the relationship between our work and \cite{blandhol_when_2022} and generalizes some of the results in the latter.

Our work is also closely related to \cite{andrews_structural_2023}, which studies the causal consistency of nonlinear IV estimators under misspecification. They show that IV estimators with instruments that do not satisfy a condition known as ``strong exclusion'' cannot approximately describe the effects of the endogenous variable $D$. The IV specifications we study do not have access to excluded variables, and therefore do not satisfy strong exclusion. We differ from \cite{andrews_structural_2023} by showing a more significant failure of the causal interpretability of IV estimators using product instruments even when imposing homogeneous, linear treatment effects, meaning the analyst's model satisfies causally correct specification.

Researchers often justify their models using economic reasoning, which can lead them to believe their model is correctly specified. This motivates \cite{gao_iv_2023}, who introduce a causal estimator that does not depend on excluded covariates but is sensitive to the functional form chosen by the analyst. If the researcher observes the true structural function $Y(D, X, \epsilon)$, or if the model proposed by the researcher $\tilde{Y}(D, X, \tilde{\epsilon}_{i})$ matches the true DGP, then the estimator from \cite{gao_iv_2023} and the linear IV specifications we study both remain unbiased. Our work focuses on cases where the researcher does not know the true structural function, so their model may be mildly misspecified. 

The remainder of the paper proceeds as follows. Section \ref{sec:model} introduces the assumptions of our model, which includes the IV conditions assumed by the analyst and the true structural function. Section \ref{sec:misspec} establishes our main result and shows that IV estimators that rely on instruments constructed from covariates can be arbitrarily biased. Finally, Section~\ref{sec:app} shows our quantitative results and illustrates the wide range of estimates produced by IV estimators that rely on product instruments across plausible DGPs. We reserve our theoretical proofs for the appendix, with the main text focusing on the key results.

\section{Structural model}\label{sec:model}

A analyst observes variables $(Y_{i}, D_{i}, X_{i}, Z_{i})$ for units $i \in \{1, \dots, n\}$. All variables are finite-dimensional, and $Y_{i} \in \mathbb{R}$. To illustrate the possibility of misspecification, we introduce a \textit{true model} that is consistent with the data generating process (DGP) and summarizes the causal relationships between the observed variables, and the \textit{analyst's} model which may rule out the true DGP. 

\subsection{True model}

Under the true model, which is consistent with the true data-generating process, the observed outcome satisfies $Y_{i} = Y(D_{i}, X_{i}, \epsilon_{i})$ for a potential outcome function $d \mapsto Y(d, X_{i}, \epsilon_{i}) \in \mathcal{Y} \subseteq \mathbb{R}$, covariates $X_{i} = (X_{i, 1}, \dots, X_{i, J})^{\top} \in \mathcal{X} \subseteq \mathbb{R}^{\text{dim}(X) \times 1}$, and unobserved heterogeneity $\epsilon_{i}$. The observed endogenous variable $D_{i}$ satisfies $D_{i} = D(X_{i}, Z_{i}, \nu_{i})$ for $D(X_{i}, Z_{i}, \nu_{i}) \in \mathcal{D} \subseteq \mathbb{R}$, a potential endogenous variable function,  unobserved heterogeneity $\nu_{i}$, and an excluded exogenous variable $Z_{i} \in \mathcal{Z} \subseteq \mathbb{R}^{\text{dim}(Z) \times 1}$. We impose one additional assumption on the true model in line with instrumental variables (IV) models.

\begin{assumption}[Excludability]\label{assum:exclude}
    The potential outcome $Y(\cdot)$ is independent of the excluded exogenous variable $Z$,
    $$Y_{i} = Y(D_{i}, X_{i}, Z_{i}, \epsilon_{i}) = Y(D_{i}, X_{i}, \epsilon_{i})$$
\end{assumption}

Under Assumption~\ref{assum:exclude}, we can write $Y_{i} = Y(D_{i}, X_{i}, \epsilon_{i})$. 

\begin{definition}[Linear IV Estimation]\label{def:IV}
    Let $f(X, Z)$ be a choice of instrumental variable for $f: \mathcal{X} \times \mathcal{Z} \mapsto \mathbb{R}$. Then, the two-stage least squares (2SLS) IV estimand $\theta_{IV}$ is given by,
    $$\theta_{IV} = \frac{\E{Yf_{\perp}}}{\E{Df_{\perp}}} = \frac{\Cov{Y, f_{\perp}}}{\Cov{D, f_{\perp}}},$$
    where $f_{\perp}$ are the residuals from a projection of the instrumental variable $f(X, Z)$ onto $X$,
    $$f_{\perp} = f(X, Z) - X^{\top}\E{XX^{\top}}^{-1}\E{Xf(X, Z)}.$$
\end{definition}

Definition~\ref{def:IV} establishes the linear IV estimand to estimate the causal effect of the treatment $D$ on the outcome $Y(D_{i}, X_{i}, \epsilon_{i})$ under a linear model. This is the estimand computed by linear method of moments and two-stage least squares estimators for an instrumental variable $f(X, Z)$ when there is a single endogenous treatment $D$.

\subsection{Analyst's model}

The analyst’s model is a special case of the true model. Under the analyst’s model, the
causal effects of the endogenous variable and included variables on the outcome are governed
by a parameter $\tilde{\theta} \in \mathbb{R}$. Here, we focus on the case where the analyst proposes and estimates a linear model with homogeneous treatment effects,
\begin{equation}\label{eq:analyst}
Y_{i} = \tilde{Y}(D_{i}, X_{i}, \tilde{\epsilon}_{i}) = \tilde{\alpha} + \tilde{\theta} D_{i} + \tilde{\pi}^{\top} X_{i} + \tilde{\epsilon}_{i}.
\end{equation}

The analyst would like to estimate the causal effect of $D$ on the outcome $Y$ in Equation~\ref{eq:analyst}, but is unwilling to assume unconfoundedness ($D \indep \tilde{\epsilon} \mid X$). Therefore, to use IV estimation, the analyst constructs an instrument $f(X, Z)$ under two assumptions sufficient to estimate $\tilde{\theta}$.

\begin{assumption}[Exogeneity] \label{assum:exogeneity}
    The unobserved heterogeneity $(\tilde{\epsilon}_{i})$ in $Y_{i}$ is uncorrelated ($\perp$), or orthogonal, to the instrumental variable $f(X_{i}, Z_{i})$ and covariates $X_{i}$,
    $$f(X, Z), X \perp \tilde{\epsilon}.$$
\end{assumption}

Assumption~\ref{assum:exogeneity} says that the instrument $f(X, Z)$ and included exogenous covariates $X$ must be uncorrelated ($\perp$) with the unobserved heterogeneity $\tilde{\epsilon}$ in the outcome under their proposed model.

\begin{assumption}[Relevance] \label{assum:non-trivial}
    The instrument $f(X, Z)$ is a random variable such that $\E{D \mid f(X, Z) = w}$ is a non-trivial function of $w$, so $\Cov{D, f(X, Z)} \not= 0$.
\end{assumption}

Assumption~\ref{assum:non-trivial} is necessary for $\theta_{IV}$ in Definition~\ref{def:IV} to be defined, so the denominator is not equal to zero. Otherwise, the treatment effect $\tilde{\theta}$ is unidentified using linear IV estimation.

In this paper, we focus on the case where the analyst does not have access to an excluded exogenous variable $Z$ to create an instrumental variable $f(X, Z)$ that is a nontrivial function of $Z$. Therefore, the analyst chooses $f(X, Z)$ to be a nonlinear function of only covariates $f(X, Z) = \tilde{f}(X) \in \mathbb{R}$ to use IV estimation.

\begin{lemma}[Unsaturated Covariates] \label{lemma:variance}
   For any function $f(X, Z) = \tilde{f}(X)$ only of covariates, $\Var{f_{\perp}} \not= 0$ if and only if $\tilde{f}(X)$ is not linear in $X$ for $f_{\perp} = \tilde{f}(X) - X^{\top}\E{XX^{\top}}^{-1}\E{Xf(X, Z)}$.
\end{lemma}

Lemma~\ref{lemma:variance} makes it clear why the analyst's instrumental variable $\tilde{f}(X)$ must be nonlinear in $X$ to use the IV estimator from Definition~\ref{def:IV}. If $\tilde{f}(X)$ is linear in each element of $X$, then the linear IV estimand $\theta_{IV}$ becomes unidentified because the denominator $\Cov{D,f_{\perp}} = \Cov{D, 0} = 0$.

\begin{lemma}[Trivial Exogeneity]\label{lemma:unneed}
    Suppose the analyst believes a stronger form of conditional exogeneity, $\E{\tilde{\epsilon} \mid X} = 0$. Then, any function of covariates $\tilde{f}(X)$ satisfies exogeneity,
    $$\tilde{f}(X) \perp \tilde{\epsilon}.$$
\end{lemma}

Lemma~\ref{lemma:unneed} demonstrates that an instrument based solely on included exogenous covariates automatically satisfies Assumption~\ref{assum:exogeneity} (exogeneity) for IV estimation when the analyst assumes a stronger exogeneity condition for these covariates. Some researchers using instruments of the form \(f(X, Z) = \tilde{f}(X)\) argue that such instruments ``require weaker identifying assumptions than using [\(X_1\) or \(X_2\)] alone'' \cite{bettinger_virtual_2017}. By including \(X_i\) in the second-stage regression, they ``control for the direct effect of [\(X_i\)] that may be related to [\(Y_i\)]'' \cite{aladangady_housing_2017}. If the covariates \(X\) included in the second-stage equation \(\tilde{Y}(D_i, X_i, \tilde{\epsilon}_i)\) satisfy conditional exogeneity (\(\mathbb{E}[\tilde{\epsilon} \mid X] = 0\)), then a nonlinear function of \(X\) can be used as an instrument without needing to independently verify its exogeneity under Assumption~\ref{assum:exogeneity}.

The following proposition shows that this estimand corresponds to the causal effect $\tilde{\theta}$ of $D_{i}$ on $Y(D_{i}, X_{i}, \epsilon_{i})$ under the analyst's model \textit{if} the true data generating process follows the model assumed by the analyst.

\begin{claim}[IV Validity] \label{clm:validity}
Assume the true data generating process $Y(D_{i}, X_{i}, \epsilon_{i})$ follows the analyst's constant linear effects model in Equation~\ref{eq:analyst} and Assumption~\ref{assum:exogeneity}, 
$$Y_{i} = Y(D_{i}, X_{i}, \epsilon_{i}) = \tilde{Y}(D_{i}, X_{i}, \tilde{\epsilon}_{i}).$$
and the analyst uses an instrument $f(X, Z) = \tilde{f}(X)$ that is a nonlinear function of covariates. Then, under Assumption~\ref{assum:non-trivial} (relevance), the IV estimand is equal to the causal effect of $D_{i}$ on $Y(D_{i}, X_{i}, \epsilon_{i})$,
$$\theta_{IV} = \frac{\Cov{Y, f_{\perp}}}{\Cov{D, f_{\perp}}} = \tilde{\theta},$$
for $f_{\perp} = \tilde{f}(X) - X^{\top}\E{XX^{\top}}^{-1}\E{Xf(X, Z)}$.
\end{claim}

Claim~\ref{clm:validity} shows that \textit{if} the analyst's model is correctly specified and their constructed instrument $f(X, Z) = \tilde{f}(X)$ satisfies Assumption~\ref{assum:non-trivial} (relevance), IV estimators $\theta_{IV}$\footnote{Note that $\theta_{IV}$ is implicitly indexed by $\tilde{f}(X)$} are consistent for the treatment effect. \cite{gao_iv_2023} show that even for nonlinear models, if the analyst's model matches the true DGP, the causal effect $\theta$ of $D$ on $Y$ can be recovered locally. We next move to analyze what happens when the analyst's model is not correctly specified.
\section{Main result: Non-interpretability}\label{sec:misspec}

When the analyst estimates a linear model like Equation~\ref{eq:analyst}, any nonlinear effect of the covariates $X$ on the outcome $Y$ is captured in the unobserved heterogeneity $\tilde{\epsilon}$. This means that IV estimators that rely on moment conditions interacting $\tilde{\epsilon}$ and the covariates $X$ will be sensitive to nonlinearity in the true structural function $Y(D_{i}, X_{i}, \epsilon_{i})$.

In Theorem~\ref{thm:non-interpret}, we illustrate this sensitivity using a toy case where the true DGP imposes linear, homogeneous treatment effects. We show that if the analyst uses an instrumental variable $f(X, Z) = \tilde{f}(X)$ that is only a function of covariates, the bias of the IV estimand $\theta_{IV}$ can be arbitrarily bad.

\begin{theorem}[Model bias]\label{thm:non-interpret}
Assume the analyst chooses an instrument $f(X_{i}, Z_{i}) = \tilde{f}(X_{i})$ that is a nonlinear function of covariates.\footnote{Recall from Lemma~\ref{lemma:unneed}, that this is necessary for $\theta_{IV}$ to be well-defined.} If Assumption~\ref{assum:non-trivial} holds, then, for any true effect $\theta \in \mathbb{R}$, and any level of potential bias $\rho \in \mathbb{R}$, there exists some continuous function $h: \mathcal{X} \mapsto \mathbb{R}$ such that if the true DGP satisfies,
$$\qquad \qquad \qquad \qquad \qquad Y_{i} = Y(D_{i}, X_{i}, \epsilon_{i}) = \alpha + \theta D_{i} + h(X_{i}) + \epsilon_{i}, \qquad \qquad \qquad \qquad (\epsilon_{i} \indep D_{i}, X_{i}).$$
Then, the bias of the linear IV estimand for the treatment effect is given by $\theta - \theta_{IV} = \rho$.
\end{theorem} 

The proof of Theorem~\ref{thm:non-interpret} is constructive -- we propose an adversarial $h(\cdot)$ which has a nonlinear term proportional to $\rho$. As we raise $\rho$, increasing the nonlinearity of the function $h(\cdot)$, the bias of the IV estimand $\theta_{IV}$ for the true treatment effect $\theta$ increases. This demonstrates the close relationship between the nonlinearity of the true structural function $Y(D_{i}, X_{i}, \epsilon_{i})$ in $X_{i}$ and the bias of the linear IV estimand. 

\begin{remark}[Causal consistency]\label{rmk:causal_consist}
    The adversarial bias result in Theorem~\ref{thm:non-interpret} is closely related to Proposition 3 of \cite{andrews_structural_2023}, which establishes that strong exclusion of the instrument $f(X, Z)$ is both necessary and sufficient for the IV estimand $\theta_{IV}$ to be approximately causally consistent. In the class of DGPs considered in Theorem~\ref{thm:non-interpret}, the analyst’s model $\tilde{Y}(D_i, X_i, \tilde{\epsilon}i)$ is “causally correctly specified,” in the sense that the structural parameter $\theta$ can be described by the estimate $\hat{\theta}_{IV}$ they report. However, because the analyst uses only covariates $(X_{1}, X_{2})$ and not the excluded variable $Z$ in their instrument $f(X, Z)$, the causal interpretation of $\hat{\theta}_{IV}$ can be arbitrarily misleading.
\end{remark}

\begin{corollary}[Product instrument]\label{cor:product}
    Let Assumption~\ref{assum:non-trivial} hold, and assume the analyst uses an instrument $f(X, Z) = X_{1} \cdot X_{2}$ that is a product of covariates $(X_{1}, X_{2})$ to estimate the IV model in Equation~\ref{eq:analyst}. If the true DGP satisfies,
    $$\qquad \qquad \qquad  Y_{i} = Y(D_{i}, X_{i}, \epsilon_{i}) = \alpha + \theta D_{i} + \pi_{1} X_{i,1} + \pi_{2} X_{i,2} + \rho X_{i, 1} X_{i, 2} + \epsilon_{i}, \qquad \qquad (\epsilon_{i} \indep D_{i}, X_{i}).$$
    Then, the bias of $\theta_{IV}$ for the treatment effect is given by,
    $$\theta - \theta_{IV} = \rho \times \frac{\E{D \left( X_{1}X_{2} - X^{\top}\E{XX^{\top}}\E{X(X_{1} \cdot X_{2})}\right)}}{\Var{X_{1}X_{2} - X^{\top}\E{XX^{\top}}\E{X(X_{1} \cdot X_{2})}}}$$
\end{corollary}

For the product instrument $\tilde{f}(X) = X_{1} \cdot X_{2}$ used in the papers cited in the introduction, Corollary~\ref{cor:product} characterizes the bias of the IV estimand for a single linear DGP. It shows that the bias is proportional to the coefficient on the interaction term $X_{1} \cdot X_{2}$. As the structural function includes increasingly nonlinear interactions, the bias in the IV estimand $\theta_{IV}$ grows accordingly. This result highlights that sizable violations of the exclusion restriction—leading to substantial bias—can arise when the interaction of excluded covariates strongly affects the outcome. For instance, if $X_{1}$ is a binary variable such as gender, then a large coefficient on $X_{1} \cdot X_{2}$ in the structural function $Y(D_i, X_i, \epsilon_i)$ suggests that the relationship between $X_{2}$ and the outcome differs markedly between men and women. To visualize this connection, Figure~\ref{fig:bias} plots IV estimates $\hat{\beta}_{IV}$ against data generated from the model in Corollary~\ref{cor:product}, illustrating how this misspecification undermines the interpretability of the IV estimand.

\begin{figure}[H]
\vspace{5pt}
\begin{centering}
\caption{\label{fig:bias} Simulation of IV estimates for the DGP $Y_{i} = \alpha + \theta D_{i} + \pi_{1} X_{i,1} + \pi_{2} X_{i,2} + \rho X_{i, 1} X_{i, 2} + \epsilon_{i}.$}
     \begin{subfigure}[b]{0.49\textwidth}
         \centering
         \includegraphics[width=\textwidth]{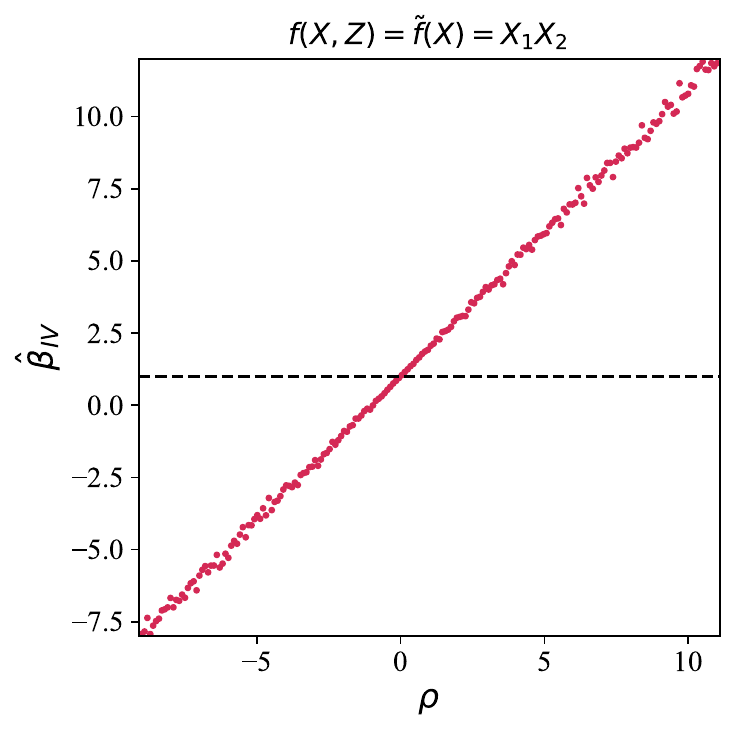}
         \vspace{-10pt}
         \caption{Instrumental variable that is nonlinear function of included variables $(X_{1}, X_{2})$}
     \end{subfigure}
     \hfill
     \begin{subfigure}[b]{0.49\textwidth}
         \centering
         \includegraphics[width=\textwidth]{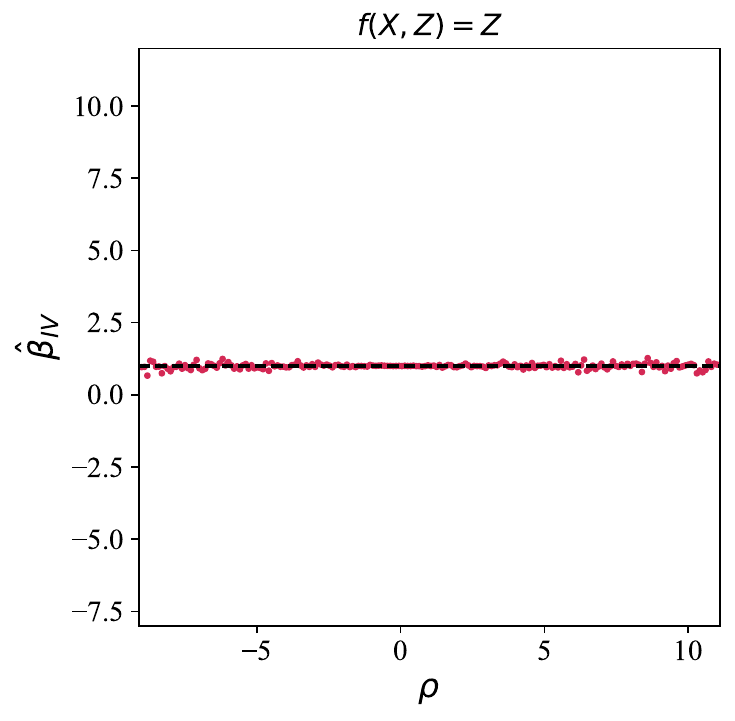}
         \vspace{-10pt}
         \caption{Instrumental variable that is function of excluded variable $Z$}
     \end{subfigure}
\par
\end{centering}
{\small{}Notes: Each dot represents an estimate of $\theta_{IV}$ from Definition~\ref{def:IV} using 2SLS with 1,000 generated observations. The dashed line indicates the true effect $\theta = 1$ of $D$ on $Y$ in the simulations.}
\vspace{-0.5pt}
\end{figure}

Examining panel (a) of Figure~\ref{fig:bias} illustrates the sensitivity of these IV estimates to misspecification in the true relationship between $X$ and $Y$. As Corollary~\ref{cor:product} implies, we see that by increasing the coefficient on the interaction term $X_{1} \cdot X_{2}$ in the true structural equation $Y(D_{i}, X_{i}, \epsilon_{i})$, we can inflate the bias of the IV estimate $\hat{\theta}$ while holding the true treatment effect $\theta$ constant. In contrast, the consistency of the IV estimates in panel (b) underscores the robustness of the IV estimator against misspecification in the true relationship between $X$ and $Y$ for partially linear DGPs when using a product instrument.

\begin{remark}[OVB]\label{rmk:ovb}
    The bias from Theorem~\ref{thm:non-interpret} can be interpreted as a form of omitted variable bias (OVB), resulting from the exclusion of the nonlinear effect of the covariates $X$ on the outcome $Y$ from the analyst's model.
\end{remark}

\begin{remark}[Rich models]\label{rmk:rich}
    Theorem~\ref{thm:non-interpret} shows that the linear IV estimand can have arbitrary bias even for simple DGPs with homogeneous linear effects. The proof further implies that estimates can also be biased for richer structural functions $Y(X_{i}, D_{i}, \epsilon_{i})$ that are nonlinear in $X$.
\end{remark}

Remark~\ref{rmk:ovb} gives an alternative interpretation of the bias for IV estimators. Remark~\ref{rmk:rich} underlines the generality of Theorem~\ref{thm:non-interpret}: for IV estimands that rely on instruments of covariates $\tilde{f}(X)$, if the analyst's model does not capture how the covariates $X$ enter the true structural function $Y(D_{i}, X_{i}, \epsilon_{i})$, then the IV estimand can be very biased.

We have shown that the practice of using instruments constructed from nonlinear functions of covariates $f(X, Z) = \tilde{f}(X)$ can misbehave in interesting ways theoretically. In Section~\ref{sec:app}, we demonstrate that this behavior is important in empirical work.

\section{Application}\label{sec:app}

Researchers often have strong economic reasoning to justify their proposed model and may believe that their model matches the true DGP. In these cases, Claim~\ref{clm:validity} (IV validity) shows that the IV estimand $\theta_{IV}$ gives the average causal effect of $D$ on $Y(D_{i}, X_{i}, \epsilon_{i})$. 

However, it is also often the case in the literature that non-nested models are used to estimate the same treatment effect. In these instances, every model cannot match the true DGP, and therefore some of these models must be misspecified. 

For example, recall from Section~\ref{sec:intro} that \cite{aladangady_housing_2017} uses IV estimation to investigate the effect of log changes in house price $D$ on log consumer expenditure $Y$ in the United States using a model of the form,
$$Y_{i} = \tilde{Y}(D_{i}, X_{i}, \tilde{\epsilon}_{i}) = \tilde{\alpha}_{0} + \tilde{\alpha}_{1} D_{i} + \tilde{\pi}^{\top} X_{i} + \tilde{\epsilon}_{i},$$
for $X\in \mathbb{R}^{2}$. Real interest rates $X_{1}$ and housing supply elasticity measures $X_{2}$ as a function of future expectations $X_{3}$ are included as covariates.

Similarly, \cite{disney_house_2010} estimates the effect of log house price growth $D$ on log consumption $Y$ in the United Kingdom, but includes interest rates $X_{1}$ and future expectations $X_{3}$ in their life-cycle model for expenditure. They use ordinary least-squares to estimate the model,
$$Y_{i} = \tilde{Y}(D_{i}, X_{i}, \tilde{\epsilon}_{i}) = \tilde{\alpha}_{0} + \tilde{\alpha}_{1} D_{i} + \tilde{\pi}^{\top} \log(X_{i}) + \tilde{\epsilon}_{i},$$
for $X \in \mathbb{R}^{7}$,\footnote{The authors also include gross household income, passive savings, estimated unanticipated house growth, and household characteristics as covariates.} where we write $\log(X)$ to be $\log(X_{i}) = (\log(X_{1}), \log(X_{3}), \dots)$. 

Both authors justify their models with economic intuition. Still, they do not agree on the functional form of consumer expenditure, or the way the covariates $X$ enter the model $\tilde{Y}(D_{i}, X_{i}, \epsilon_{i})$. Theorem~\ref{thm:non-interpret} (model bias) shows that if the author's model \textit{does not} match the true structural function, IV estimators that use an instrument constructed from covariates can have a high level of bias. Therefore, it is important to understand the sensitivity of these estimates to how the covariates $X$ affect the outcome $Y$ in cases when there are several reasonable choices to model the DGP.

Table~\ref{tab:models} gives additional examples of competing models used in the literature. For each of these examples, we compare: on one hand, (i) a linear IV model that relies on the product instrument $\tilde{f}(X) = X_{1} \cdot X_{2}$; and on the other hand (ii) a richer nonlinear model proposed in that literature to estimate the same treatment effect. For each pair of papers, the linear IV model is written on the top row and the non-nested alternative model is given on the bottom row.

\vspace{1em}
\renewcommand{\arraystretch}{1.5}
\begin{table}[H]
\small
\begin{center}
\caption{\centering \label{tab:models} Alternative choices for $\tilde{Y}(D_{i}, X_{i}, \tilde{\epsilon}_{i})$}
\begin{tabular}{c@{\hskip 1em}c@{\hskip 2em}cc}
Paper & Model & Outcome $Y$ & Endo. Variable $D$ \\[0.5em]
\hline
\cite{aladangady_housing_2017} & $Y_{i} = \tilde{\alpha}_{0} + \tilde{\alpha}_{1}D_{i} + \tilde{\pi}^{\intercal} X_{i} + \tilde{\epsilon}_{i}$ & Log consumption & Log house wealth\\[0.5em]
\cite{disney_house_2010} & $Y_{i} = \tilde{\alpha}_{0} + \tilde{\alpha}_{1}D_{i} + \tilde{\pi}^{\intercal} \log(X_{i}) + \tilde{\epsilon}_{i}$ & Log consumption & Log house wealth\\[0.5em]
\cmidrule(lr){1-1} & \\[-1.75em]
\cite{munshi_networks_2016} & $Y_{i} = \tilde{\alpha}_{0} + \tilde{\alpha}_{1}D_{i} + \tilde{\pi}^{\intercal} X_{i}  + \tilde{\epsilon}_{i}$ & Migration indicator & Level of wealth \\[0.5em]
\cite{abramitzky_limits_2008} & $Y_{i} = \Phi(\tilde{\alpha}_{0} + \tilde{\alpha}_{1}D_{i} + \tilde{\pi}^{\intercal} X_{i}) + \tilde{\epsilon}_{i}$ & Migration indicator & Level of wealth \\[0.5em]
\cmidrule(lr){1-1} & \\[-1.75em]
\cite{helmers_board_2017} & $Y_{i} = \tilde{\alpha}_{0} + \tilde{\alpha}_{1}D_{i} + \tilde{\pi}^{\intercal} X_{i} + \tilde{\epsilon}_{i}$ & Patent count & Network size \\[0.5em]
\cite{chang_when_2006} & $Y = \exp{\tilde{\alpha}_{0} + \tilde{\alpha}_{1}D_{i} + \tilde{\pi}^{\intercal} X_{i} } + \tilde{\epsilon}_{i}$ & Patent count & Network size \\[0.5em]
\cmidrule(lr){1-1} & \\[-1.75em]
\cite{liu_housing_2023} & $Y_{i} = \tilde{\alpha}_{0} + \tilde{\alpha}_{1}D_{i} + \tilde{\pi}^{\intercal} X_{i} + \tilde{\epsilon}_{i}$ & Indicator of child & Value of home \\[0.5em]
\cite{clark_effect_2019} & $Y_{i} = \frac{\exp{\tilde{\alpha}_{0} + \tilde{\alpha}_{1}D_{i} + \tilde{\pi}^{\intercal} X_{i}}}{1 + \exp{\tilde{\alpha}_{0} + \tilde{\alpha}_{1}D_{i} + \tilde{\pi}^{\intercal}X_{i}}} + \tilde{\epsilon}_{i}$ & Indicator of child & Value of home \\[0.5em]
\hline
\end{tabular}
\end{center}
\bigskip
{\small{} Notes: Of the papers cited in the introduction that rely on IV estimation using the product instrument $\tilde{f}(X) = X_{1} \cdot X_{2}$, Table~\ref{tab:models} gives an overview of the alternative models used among papers cited in their literature review.}
\end{table}

Table~\ref{tab:models} shows that cases often arise when there is no consensus about the functional form of the true structural function $Y(D_{i}, X_{i}, \epsilon_{i})$. Therefore, it is likely that the true DGP is nonlinear in the covariates $X$.

To assess the sensitivity of linear IV models that use product instruments, Table~\ref{tab:simulation} performs a semi-synthetic exercise for each pair of papers in Table~\ref{tab:models}. Specifically, for each pair of papers, we interpret the richer, nonlinear specification as the true structural equation and calibrate the DGP to match that model. We restrict one degree of freedom for the parameters $(\alpha^{\intercal}, \pi^{\intercal})$ by forcing $\theta = \E{\pdv{}{D_{i}} Y(D_{i}, X_{i}, \epsilon_{i})} = 1$. We then use 2SLS to estimate $\hat{\theta}$ in the linear model IV model from Equation~\ref{eq:analyst}.

For example, using the structural model offered by \cite{disney_house_2010} we construct synthetic draws of $(Y_{i}, D_{i}, X_{i})$ such that $\theta = \alpha_{1} = 1$. Then, we fit the linear IV model used by \cite{aladangady_housing_2017} and compare how the IV estimate $\hat{\theta}$ compares to $\theta$ in the synthetic structural model. Full details are given in Appendix~\ref{sec:sim}.

\subsection{Results}

Table~\ref{tab:simulation} shows estimates of $\hat{\theta}$ from Equation~\ref{eq:analyst} for the estimand $\theta_{IV}$ in Definition~\ref{def:IV} using linear IV estimation. 

\renewcommand{\arraystretch}{1.5}
\begin{table}[H]
\begin{center}
\caption{\centering \label{tab:simulation} Estimates of $\hat{\theta}$ from Equation~\ref{eq:analyst} using 2SLS}
\begin{tabular}{l@{\hskip 1em}l@{\hskip 2em}l@{\hskip 2em}c}
 Paper & Model: $\tilde{Y}(D_{i}, X_{i}, \tilde{\epsilon}_{i}) = $ & DGP: $Y(D_{i}, X_{i}, \epsilon_{i}) = $  & $\hat{\theta}$ \\
\hline\\[-1.75em]
 \shortstack{\cite{aladangady_housing_2017} \\ \phantom{a}} & \shortstack{$\tilde{\alpha} + \tilde{\theta}D_{i} + \tilde{\pi}^{\intercal} X_{i} + \tilde{\epsilon}_{i}$ \\ \phantom{a}} &  \shortstack{$\alpha_{0} + \alpha_{1}D_{i} + \pi^{\intercal} \log(X_{i}) + \epsilon_{i}$ \\ \phantom{a}} & \shortstack{$-1.16$ \\ $\emph{(-2.29, -0.02)}$} \\[0.25em]
 \shortstack{\cite{munshi_networks_2016} \\ \phantom{a}}    &  \shortstack{$\tilde{\alpha} + \tilde{\theta}D_{i} + \tilde{\pi}^{\intercal} X_{i} + \tilde{\epsilon}_{i}$ \\ \phantom{a}} & \shortstack{$\Phi\left(\alpha_{0} + \alpha_{1}D_{i} + \pi^{\intercal} X_{i}\right) + \epsilon_{i}$ \\ \phantom{a}} & \shortstack{$-2.42$ \\ $\emph{(-3.77, -1.06)}$} \\[0.25em]
 \shortstack{\cite{helmers_board_2017} \\ \phantom{a}}      & \shortstack{$\tilde{\alpha} + \tilde{\theta}D_{i} + \tilde{\pi}^{\intercal} X_{i} + \tilde{\epsilon}_{i}$ \\ \phantom{a}} &  \shortstack{$\exp{\alpha_{0} + \alpha_{1}D_{i} + \pi^{\intercal} X_{i}} + \epsilon_{i}$ \\ \phantom{a}} & \shortstack{$0.24$ \\ $\emph{(0.18, 0.30)}$} \\[0.25em]
 \shortstack{\cite{liu_housing_2023} \\ \phantom{a}}  & \shortstack{$\tilde{\alpha} + \tilde{\theta}D_{i} + \tilde{\pi}^{\intercal} X_{i} + \tilde{\epsilon}_{i}$ \\ \phantom{a}}  & \shortstack{$\frac{\exp{\alpha_{0} + \alpha_{1}D_{i} + \pi^{\intercal} X_{i}}}{1 + \exp{\alpha_{0} + \alpha_{1}D_{i}+ \pi^{\intercal} X_{i}}} + \epsilon_{i}$ \\ \phantom{a}} & \shortstack{$0.35$ \\ $\emph{(0.30, 0.41)}$} \\[0.25em]
\hline
\end{tabular}
\end{center}
\medskip
{\small Notes: Each row estimates $\hat{\theta}$ from Equation~\ref{eq:analyst} using the IV model from the corresponding paper. Data is generated following the richer, nonlinear model, i.e. the structural function listed under DGP. $\hat{\theta}$ reports the average estimate across 1000 simulations for the estimand $\theta_{IV}$ in Defintion~\ref{def:IV}. 95\% confidence intervals for $\hat{\theta}$ across simulations are reported in parenthesis. For each DGP, $(\alpha^{\intercal}, \pi^{\intercal})$ are chosen while maintaining the condition $\theta = \E{\pdv{}{D_{i}} Y(D_{i}, X_{i}, \epsilon_{i})} = 1$. Details of this calculation are at the end of Appendix~\ref{sec:sim}.}
\end{table}

In Table~\ref{tab:simulation}, we observe that for all four pairs of papers, a 95\% confidence interval for the estimated treatment effect $\hat{\theta}$ does not include the ground truth of $\theta = 1$, i.e. the average partial treatment effect (APE) in the true structural function $Y(D_{i}, X_{i}, \epsilon_{i})$. For instance, IV estimates tend be negative when using the models from \cite{aladangady_housing_2017}, where the DGP is linear with log-transformed covariates $X$. Meanwhile, for the models used by \cite{helmers_board_2017} and \cite{liu_housing_2023}, where the true structural function follows an exponential and logit design, IV estimates based on instruments that are functions of covariates tend to be close to zero, underestimating the APE. This illustrates how the bias of the IV estimand can be very positive or negative when using product instruments, depending on the underlying DGP.

These findings underscore the sensitivity of these IV estimates $\hat{\theta}$ to the functional form chosen to model the relationship between covariates $X$ and the outcome $Y$ and illustrate their bias under sensible choices for the DGP. Even for reasonable choices of $Y(D_{i}, X_{i}, \epsilon_{i})$, IV estimates of the treatment effect $\hat{\theta}$ with constructed instruments can differ substantially from the true APE and imply different causal conclusions.

Another natural question is how sensitive causal estimates are to the analyst’s modeling choices when applied to real economic data. To examine this, Table~\ref{tab:application} explores how estimates of $\theta$ vary when fitting different functional forms to data from published empirical studies. Specifically, we consider the subset of papers published in the \textit{American Economic Review} (AER) since 2014 that employ instrumental variable (IV) strategies relying on product instruments of the form $\tilde{f}(X) = X_{1} \cdot X_{2}$.\footnote{We only use papers for which the data is publicly available.} For each of these studies,\footnote{Paper (1) is \cite{rogall_mobilizing_2021} and paper (2) is \cite{munshi_networks_2016}.} we re-estimate the parameter $\tilde{\theta}$ using the linear IV model,
\begin{equation}\begin{aligned} \label{eq:tab_model}
    Y_{i} = \tilde{Y}(D_{i}, X_{i}, \tilde{\epsilon}_{i}) &= \tilde{\alpha}_{0} + \tilde{\theta} D_{i} + \tilde{\pi}^{\top} h(X_{i}) + \tilde{\epsilon}_{i} \ ,
\end{aligned}\end{equation}
for various nonlinear functions $h(\cdot)$ and investigate the range of estimates produced.

\vspace{1em}
\renewcommand{\arraystretch}{1.5}
\begin{table}[H]
\small
\begin{center}
\caption{\centering \label{tab:application} Estimates of $\tilde{\theta}$ from Equation~\ref{eq:tab_model}.}
\begin{tabular}{l@{\hskip 5em}c@{\hskip 3em}c@{\hskip 3em}c}
& \multicolumn{3}{c}{Paper} \\[-0.25em] 
\cmidrule(lr){2-4} & \\[-3em]
 $h(X_{i}) =$                                                                                                                          & (1)                                          & (2a)                                        & (2b)                                         \\
\hline & \\[-1.75em]
 \shortstack{$(X_{i, 1}, \dots, X_{i, q})^{\top}$ \\ \phantom{a}}                                                                                         & \shortstack{1.34 \\ \textit{(-0.30, 2.99)}}  & \shortstack{5.40 \\ \textit{(3.76, 7.05)}}    & \shortstack{-0.72 \\ \textit{(-2.37, 0.92)}}  \\[0.25em]
 \shortstack{$(\exp{X_{i, 1} - \bar{X}_{1}}, \dots, \exp{X_{i, q} - \bar{X}_{q}})^{\top}$ \\ \phantom{a}}                                                 & \shortstack{1.55 \\ \textit{(-0.16, 3.25)}}  & \shortstack{1.58 \\ \textit{(-0.03, 3.19)}}   & \shortstack{-0.73 \\ \textit{(-2.42, 0.96)}}  \\[0.25em]
 \shortstack{$(\log(1 + X_{i, 1}), \dots, \log(1 + X_{i, q}))^{\top}$ \\ \phantom{a}}                                                                     & \shortstack{5.54 \\ \textit{(-5.12, 16.21)}} & \shortstack{1.41 \\ \textit{(-0.35, 3.18)}}   & \shortstack{-0.60 \\ \textit{(-2.24, 1.04)}}  \\[0.25em]
 \shortstack{$\left(\frac{2\exp{2 X_{i, 1}} - 2\exp{-2 X_{i, 1}}}{2.5\exp{2.5 X_{i, 1}} + 2.5 \exp{- 2.5 X_{i, 1}}}, \dots\right)^{\top}$ \\ \phantom{a}} & \shortstack{0.39 \\ \textit{(-0.98, 1.76)}}  & \shortstack{0.65 \\ \textit{(-11.52, 12.81)}} & \shortstack{-5.46 \\ \textit{(-13.36, 2.43)}} \\[0.5em]
\hline
\end{tabular}
\end{center}
\medskip
{\small Notes: Each row estimates the linear IV estimand $\tilde{\theta}$ from Equation~\ref{eq:tab_model} for the specification listed under $h(X_{i}) =$ using 2SLS. We standardized each column by dividing it by the standard deviation of the linear estimate, which corresponds to the IV estimate using the authors' model $h(X_{i}) = (X_{i, 1}, \dots, X_{i, q})^{\top}$ for $q = \dim(X_{i})$. Asymptotic 95\% confidence intervals using heteroskedasticity-robust standard errors are in parenthesis.}
\end{table}

Within each paper, we find that the estimated causal effect $\tilde{\theta}$ is sensitive to the choice of functional form for $h(X_{i})$. In paper (1), for instance, a 95\% confidence interval based on the reported linear specification does not include the estimate obtained using a log-transformed specification, where $h(X_{i}) = (\log(1 + X_{i,1}), \dots, \log(1 + X_{i,q}))^{\top}$. Similarly, in paper (2a), none of the estimates from nonlinear specifications lie within the 95\% confidence interval derived from the original linear model.

These results demonstrate that the analyst’s estimate of $\tilde{\theta}$ can vary substantially depending on the functional form used to model the relationship between covariates $X$ and the outcome $Y$. Even among plausible specifications of $h(X_{i})$, the magnitude and statistical significance of the estimated treatment effect can differ markedly.

Together, this evidence highlights the lack of robustness of IV estimators based on product instruments to nonlinearities in how covariates $X$ enter the structural function $Y(D_{i}, X_{i}, \eta_{i})$. Substantive conclusions drawn from such models may be sensitive to modeling assumptions.

Overall, our results underscore the importance of functional form choices in IV models where instruments are functions of covariates, $f(X, Z) = \tilde{f}(X)$. The selected functional form can significantly influence how the treatment effect is interpreted.

\section{Recommendations for practice}

In this paper, we have provided theoretical and empirical evidence showing that product instruments, or more generally, instruments that are transformations of covariates, can be very biased. Our framework shows that even when imposing homogeneous linear effects between an endogenous variable and outcome, these IV estimates cannot be interpreted as a causal effect under mild forms of misspecification. Theoretically, we show that there exists a simple DGP with constant linear effects that can lead to unbounded bias of the IV estimand. Empirically, we illustrate the sensitivity of IV estimates that use instruments constructed from nonlinear functions of covariates.

When a researcher does not have access to an excluded, exogenous variable, we recommend caution when discussing the causal interpretability of IV estimators. Alternatively, if the researcher believes their instrumental variable $f(X, Z) = \tilde{f}(X)$ satisfies the IV assumptions in Section~\ref{sec:model}, it is important that they strongly defend their choice of model $\tilde{Y}(\cdot)$, i.e. the relationship between the outcome $Y$ and covariates $X$, to defend the validity of their estimate.

\bibliographystyle{apalike}
\bibliography{main}
\newpage

\appendix
\begin{center}
\renewcommand*{\thefootnote}{\arabic{footnote}}
\setcounter{footnote}{0} \textbf{\Large{}Appendix for}\\
\textbf{\Large{}``Constructing Instruments as a Function of Included Covariates''}{\Large\par}
\par\end{center}

\noindent \begin{center}
Moses Stewart, \emph{Harvard University}\textbf{}\footnote{E-mail: mosesstewart@g.harvard.edu}\textbf{}\\
\par\end{center}

\appendix
\newcounter{appndx}
\counterwithin{table}{appndx}
\renewcommand\thetable{\Alph{appndx}\arabic{table}}
\renewcommand{\tablename}{Appendix Table}
\setcounter{table}{0}
\counterwithin{figure}{appndx}
\renewcommand\thefigure{\Alph{appndx}\arabic{figure}}
\renewcommand{\figurename}{Appendix Figure}
\setcounter{figure}{0}
\setcounter{page}{1}
\setcounter{section}{0} 

\section{Proofs from main text}\label{sec:proofs}

\begin{proof}[Proof of Lemma~\ref{lemma:variance}]
    \hfill\newline
    To prove the forward direction, let $\Var{f_{\perp}} = \E{f_{\perp}^{2}} = 0$. This implies that $f_{\perp} = 0$ almost surely. However, since $X^{\top}\E{XX^{\top}}^{-1}\E{X\tilde{f}(X)}$ is a linear projection of $\tilde{f}(X)$ onto $X$, this implies that $\tilde{f}(X)$ is a linear function of each $X_{,j} \in X$. 
    
    Similarly to prove the backwards direction, let $\Var{f_{\perp}} = \E{f_{\perp}^{2}} \neq 0$. Since $X^{\top}\E{XX^{\top}}^{-1}\E{X\tilde{f}(X)}$ is a linear projection of $\tilde{f}(X)$ onto $X$, this implies that $f_{\perp}$ is not equal to $0$ almost surely, so it is not a linear function of each $X_{, j} \in X$. The result follows.
\end{proof}

\begin{proof}[Proof of Lemma~\ref{lemma:unneed}]
    \hfill\newline
    Using the tower property, notice that,
    
    \begin{equation*}\begin{aligned}    
    \Cov{\tilde{f}(X_{i}), \tilde{\epsilon}_{i}} = \E{\tilde{\epsilon}_{i}\tilde{f}(X_{i})} - \E{\tilde{\epsilon}_{i}}\E{\tilde{f}(X_{i})}
    = \E{\E{\tilde{\epsilon}_{i} \mid X_{i}}\tilde{f}(X_{i})}
    = 0.
    \end{aligned}\end{equation*}
\end{proof}
\begin{proof}[Proof of Claim~\ref{clm:validity}]
\hfill\newline
We assume that the true DGP matches the analyst's model, so,
$$Y_{i} = Y(D_{i}, X_{i}, \epsilon_{i}) = \alpha_{i} + D_{i}\theta + \pi^{\top}X_{i} + \epsilon_{i}.$$
Then since under Assumption~\ref{assum:exogeneity}, we have that $\Cov{\epsilon_{i}, f_{\perp}} = 0$ we can write,
    \begin{equation*}\begin{aligned}
    \theta_{IV} &= \frac{\Cov{Y, f_{\perp}}}{\Cov{D, f_{\perp}}} = \frac{\Cov{Y(D, X, \epsilon), f_{\perp}}}{\Cov{D, f_{\perp}}}\\
    &= \frac{\Cov{\alpha, f_{\perp}}}{\Cov{D, f_{\perp}}} + \frac{\Cov{\theta D, f_{\perp}}}{\Cov{D, f_{\perp}}} +  \frac{\Cov{X \pi, f_{\perp}}}{\Cov{D, f_{\perp}}} + \frac{\Cov{\epsilon, f_{\perp}}}{\Cov{D, f_{\perp}}}\\
    &= \theta \frac{\Cov{D, f_{\perp}}}{\Cov{D, f_{\perp}}} +  \pi \frac{\Cov{X, f_{\perp}}}{\Cov{D, f_{\perp}}}\\
    &= \theta .
    \end{aligned}\end{equation*}
    The assertion that $\Cov{X_{i}, f_{\perp}} = 0$ in the last line follows since $f_{\perp} = \tilde{f}(X) - X^{\top}\E{XX^{\top}}^{-1}\E{X\tilde{f}(X)}$ are the residuals from the projection of $\tilde{f}(X)$ onto $X$, and $\Cov{D, f_{\perp}} \not= 0$ follows from Assumption~\ref{assum:non-trivial} (relevance).  
\end{proof}

\begin{proof}[Proof of Theorem~\ref{thm:non-interpret}]
    \hfill\newline
    To prove Thereom~\ref{thm:non-interpret}, we will construct a continuous function $h(X_{i})$ such that the estimand $\theta_{IV} = \rho$. Let $f_{\perp} = \tilde{f}(X) - X \E{XX^{\top}}^{-1}\E{X\tilde{f}(X)}$ be the residuals from a projection of the instrumental variable $\tilde{f}(X)$ onto $X$ and $g(\rho) = \frac{\Cov{D, f_{\perp}}}{\Var{f_{\perp}}}\rho$. Consider the function,
    \begin{equation*}\begin{aligned}
    Y_{i} = Y(D_{i}, X_{i}, \epsilon_{i}) &= \alpha + \theta D_{i} + h(X_{i}) + \epsilon_{i}\\
    \text{for }\qquad h(X_{i}) &= X_{i} \pi + g(\rho) \tilde{f}(X_{i}).
    \end{aligned}\end{equation*}
    Then, we can write the IV estimand as,
    \begin{equation*}\begin{aligned}
    \theta_{IV} &= \frac{\Cov{Y, f_{\perp}}}{\Cov{D, f_{\perp}}} = \frac{\Cov{\alpha + \theta D_{i} + h(X_{i}) + \epsilon_{i}, f_{\perp}}}{\Cov{D, f_{\perp}}} \qquad \text{(I)}\\
    &= \frac{\Cov{\alpha, f_{\perp}}}{\Cov{D, f_{\perp}}} + \frac{\Cov{\theta D, f_{\perp}}}{\Cov{D, f_{\perp}}} + \frac{\Cov{X\pi, f_{\perp}}}{\Cov{D, f_{\perp}}} + \frac{\Cov{g(\rho) \tilde{f}(X_{i}), f_{\perp}}}{\Cov{D, f_{\perp}}} + \frac{\Cov{\epsilon, f_{\perp}}}{\Cov{D, f_{\perp}}}\\
    &= \theta \frac{\Cov{D, f_{\perp}}}{\Cov{D, f_{\perp}}} + \frac{\Cov{X, f_{\perp}}}{\Cov{D, f_{\perp}}}\pi + g(\rho) \frac{\Var{f_{\perp}}}{\Cov{D, f_{\perp}}} \qquad \text{(II)}\\
    &= \theta + g(\rho) \frac{\Var{f_{\perp}}}{\Cov{D, f_{\perp}}} \qquad \text{(III)} \\
    &= \rho.
    \end{aligned}\end{equation*}
    Where (I) follows because we assumed the DGP satisfies $Y = \alpha + \theta D_{i} + h(X_{i}) + \epsilon_{i}$. $\Cov{\epsilon, f_{\perp}} = 0$ in line (II) follows since we impose $\epsilon \indep X$. $\Cov{X, f_{\perp}} = 0$ in line (III) follows since $f_{\perp} = \tilde{f}(X) - X^{\top}\E{XX^{\top}}^{-1}\E{X\tilde{f}(X)}$ are the residuals from the projection of $\tilde{f}(X)$ onto $X$, and $\Cov{D, f_{\perp}} \not= 0$ follows from Assumption~\ref{assum:non-trivial} (relevance). 
\end{proof}
\section{Simulation details} \label{sec:sim}
\subsection{Data generating process}

We consider four DGPs corresponding to nonlinear models used in Table~\ref{tab:simulation}: the \textbf{probit} model (\cite{abramitzky_limits_2008}), the \textbf{exponential} model (\cite{chang_when_2006}), the \textbf{logit} model (\cite{clark_effect_2019}), and the \textbf{log-linear} model (\cite{disney_house_2010}).

For the first three models, we generate the data following
\begin{equation*}\begin{aligned}
    (D_{i} , X_{i,1}, X_{i, 2}, \epsilon_{i})^{\top} &\iidsim \Norm{\pmb{0}, \Sigma} \qquad \text{for } 
    \Sigma = \begin{bmatrix}
        \sigma^{2}_{d} & \rho & \rho & 0\\
        \rho & \sigma^{2}_{x1} & \rho & 0\\
        \rho & \rho & \sigma^{2}_{x2} & 0\\
        0 & 0 & 0 & \sigma^{2}_{\epsilon}\\
    \end{bmatrix}
\end{aligned}\end{equation*}

Using these draws, we define the outcome variable $Y_i$ as follows:
\begin{enumerate}
    \item \textbf{Probit} (\cite{abramitzky_limits_2008}): Define $p_i = \Phi(\alpha_1 D_i + X_{i1} + X_{i2})$ and draw $Y_i \sim \Bern{p_i}$.
    \item \textbf{Exponential} (\cite{chang_when_2006}): Set $Y_{i} = \exp(\alpha_1 D_i + X_{i1} + X_{i2}) + \epsilon_i$.
    \item \textbf{Logit} (\cite{clark_effect_2019}) Define $p_{i} = \frac{\exp{\alpha_{1} D_{i} + X_{i, 1} + X_{i, 2}}}{1 + \exp{\alpha D_{i} + X_{i, 1} + X_{i, 2}}}$ and draw $Y_i \sim \Bern{p_i}$.
\end{enumerate}

For the \textbf{log-linear} (\cite{disney_house_2010}) specification, we simulate a common latent factor $Z_i \sim \Norm{0,1}$ and independent Gamma random variables $G_{i,1}, G_{i,2} \sim \Gam{\mu_k, \sigma_k^2 - \rho}$, where the Gamma distribution is parameterized by its mean and variance. We then generate the variables as,
$$X_{i1} = G_{i1} + \sqrt{\rho} Z_i, \qquad
X_{i2} = G_{i2} + \sqrt{\rho} Z_i, \qquad
D_i = \delta_i + \sqrt{\rho} Z_i, \qquad \text{with } \delta_{i} \iidsim \mathcal{N}(0, \sigma_d^2 - \rho).$$
Finally, for $\epsilon_{i} \iidsim \Norm{0, \sigma^{2}_{\epsilon}}$ we define the outcome as:
$$
Y_i = D_i + \pi_1 \log X_{i1} + \pi_2 \log X_{i2} + \epsilon_i
$$

\subsection{\texorpdfstring{$\alpha$}{Alpha}-Calibration}

Let $g(\cdot)$ denote the structural function used in the DGP. To standardize the marginal effect of $D_i$, we calibrate the parameter $\alpha_1$ so that the average partial derivative of $g(\alpha_1 D_i + X_{i1} + X_{i2})$ with respect to $D_i$ equals one.

Since the sum $\alpha_1 D_i + X_{i1} + X_{i2}$ is distributed as $W \sim \mathcal{N}(0, \alpha_1^2 (\sigma_d^2 + 2\rho) + \sigma_{x1}^2 + \sigma_{x2}^2)$, we write,
$$g(\alpha_1 D_i + X_{i1} + X_{i2}) \sim g(W), \quad \text{so} \quad \pdv{}{D_i} g(\alpha_1 D_i + X_{i1} + X_{i2}) \sim \alpha_1 g'(W).$$

We therefore choose $\alpha_1$ to satisfy $\E{\alpha_1 g'(W)} = 1$ by minimizing the objective $\left( \alpha_1 - \E{g'(W)}^{-1} \right)^2$. This calibration ensures that the simulated average partial effect of the treatment variable $D_i$ is normalized across models.

\section{Weak Causality}\label{sec:blandhol}

\subsection{Nonlinear IV and LATE interpretation}

In this section, we generalize the theoretical results in \cite{blandhol_when_2022} to cover continuous random variables $(Y, D, Z, X)$ and non-linear treatment effects. We then use the results to analyze the local average treatment effects (LATE) interpretation of IV estimators of instruments $f(X, Z) = \tilde{f}(X)$ that are nonlinear functions of exogenous variables.

\begin{definition}[Weakly causal]\label{def:wc}
$\beta$ is \textit{weakly causal} if both of the following statements are true,
\begin{enumerate}[label=\alph*]
    \item If $Y(d, x, \epsilon)$ is non-decreasing in $d$ for all $d \in \mathcal{D}$ and $x \in \mathcal{X}$, then $\beta \geq 0$
    \item If $Y(d, x, \epsilon)$ is non-increasing in $d$ for all $d \in \mathcal{D}$ and $x \in \mathcal{X}$, then $\beta \leq 0$
\end{enumerate}
\end{definition}

Definition~\ref{def:wc} is a minimal requirement for an estimand $\beta$ to reflect the causal effect of $D$ on $Y$. It says that if the causal effect of the treatment has the same sign for individual, then the summary estimand $\beta$ also has that sign. An estimand $\beta$ which fails to satisfy Definition~\ref{def:wc} becomes uninterpretable as a measure of the treatment effect $D$ on the outcome $Y$. Note that this definition of weak causality based on the potential outcomes implies the definition in \cite{blandhol_when_2022}.

\begin{assumption}[Instrument monotonicity]\label{assum:instr_mono}
    For the potential endogenous variable function $D_{i} = D(X_{i}, Z_{i}, \upsilon_{i})$ and instrument $f(X_{i}, Z_{i})$ assume either,
    \begin{enumerate}[label = \alph*]
        \item $D(x_{i}, z_{i}, \nu_{i})$ weakly increases as $f(x_{i}, z_{i})$ increases for all $z \in \mathcal{Z}$ and $x \in \mathcal{X}$
        \item $D(x_{i}, z_{i}, \nu_{i})$ weakly decreases as $f(x_{i}, z_{i})$ increases for all $z \in \mathcal{Z}$ and $x \in \mathcal{X}$
    \end{enumerate}
\end{assumption}

Assumption~\ref{assum:instr_mono} requires that there is a consistent relationship between the endogenous variable $D$ and the instrument $f(X, Z)$. It says that as the observed instrument $f(X_{i}, Z_{i})$ increases, so does the endogenous variable $D_{i}$. Using assumption~\ref{assum:instr_mono}, we can arrive at proposition 4 of \cite{blandhol_when_2022} without assuming discrete random variables or constant, linear treatment effects.

\begin{proposition}[Weakly Causal]\label{prop:wc}
    Let Assumptions~\ref{assum:exclude},~\ref{assum:exogeneity},~\ref{assum:non-trivial}, and~\ref{assum:instr_mono} hold. Then, $\theta_{IV}$ is \textit{weakly causal} if and only if $\E{f_{\perp} \mid X} = 0$.
\end{proposition}

\begin{corollary}[model bias]\label{cor:interpret}
    For any instrumental variable $f(X, Z) = \tilde{f}(X)$ that is a function exclusively of included covariates $X$, $\theta_{IV}$ is not weakly causal.
\end{corollary}

Corollary~\ref{cor:interpret} gives an alternate result to Theorem~\ref{thm:non-interpret} that shows how IV estimands of instrumental variables $\tilde{f}(X)$ that are only functions of included covariates lack causal interpretability. It relies on the failure of the ``saturated covariates" condition ($\E{f_{\perp} \mid X} = 0$) for IV estimands which use an instrument constructed from a nonlinear function of included covariates. Corollary~\ref{cor:interpret} implies that these estimators cannot be interpreted as LATE.

\subsection{Proofs}

\begin{lemma}[FKG Inequality]\label{lemma:cov}
    Let $(X, \epsilon, \eta)$ be mutually independent random variables, and let $g(\cdot, \cdot): \mathbb{R}^{2} \mapsto \mathbb{R}$ and $h(\cdot, \cdot): \mathbb{R}^{2} \mapsto \mathbb{R}$ be non-decreasing functions in their first arguments. Then, it follows that $\Cov{g(X, \epsilon), h(X, \eta)} \geq 0$. Similarly, if $g(\cdot, \cdot)$ is non-decreasing and $h(\cdot, \cdot)$ is non-increasing, then $\Cov{g(X, \epsilon), h(X, \eta)} \leq 0$.
\end{lemma}

\begin{proof}[Proof of Lemma~\ref{lemma:cov}]
    \hfill\newline
    First, we show the result for non-decreasing functions $g(\cdot)$ and $h(\cdot)$.
    
    Define $X_{1}, X_{2} \iidsim X$, $\eta_{1}, \eta_{2} \iidsim \eta$, and $\epsilon_{1}, \epsilon_{2} \iidsim \epsilon$. Then notice that using the tower property,
    \begin{equation*}\begin{aligned} 
    &\E{(g(X_{1}, \epsilon_{1}) - g(X_{2}, \epsilon_{2}))(h(X_{1}, \eta_{1}) - h(X_{2}, \eta_{2}))}\\
    &= \E{ \E{(g(X_{1}, \epsilon_{1}) - g(X_{2}, \epsilon_{2}))(h(X_{1}, \eta_{1}) - h(X_{2}, \eta_{2})) \mid X_{1} > X_{2}}}\\
    &\geq \E{ \E{(g(X_{2}, \epsilon_{1}) - g(X_{2}, \epsilon_{2}))(h(X_{2}, \eta_{1}) - h(X_{2}, \eta_{2})) \mid X_{1} > X_{2}}} \qquad \text{(I)}\\
    &= 0 .
    \end{aligned}\end{equation*}
    Where (I) follows from the monotonicity inequality, which says if the random variable $Y_{1} \geq Y_{2}$ then $\E{Y_{1}} \geq \E{Y_{2}}$. Noticing that $g(\cdot, \cdot)$ and $h(\cdot, \cdot)$ are non-decreasing in their first argument and letting $g(X_{1}, \epsilon_{1}), h(X_{1}, \eta_{1}) = Y_{1}$ and $g(X_{2}, \epsilon_{2}), h(X_{2}, \eta_{2}) = Y_{2}$ gives the result. Next, we write,
    \begin{equation*}\begin{aligned}
        &\E{(g(X_{1}, \epsilon_{1}) - g(X_{2}, \epsilon_{2}))(h(X_{1}, \eta_{1}) - h(X_{2}, \eta_{2}))} \\
        &= \E{g(X_{1}, \epsilon_{1})h(X_{1}, \eta_{1})} - \E{g(X_{1}, \epsilon_{1})}\E{h(X_{2}, \eta_{2})} - \E{g(X_{2}, \epsilon_{2})}\E{h(X_{1}, \eta_{1})} + \E{g(X_{2}, \epsilon_{2})h(X_{2}, \eta_{2})}\\
        &= 2\E{g(X, \epsilon)h(X, \eta)} - 2\E{g(X, \epsilon)}\E{h(X, \eta)}\\
        &= 2\Cov{g(X, \epsilon), h(X, \eta)} .\\
    \end{aligned}\end{equation*}
    So it follows that,
    $$\Cov{g(X, \epsilon), h(X, \eta)} = \frac{1}{2} \E{(g(X_{1}, \epsilon_{1}) - g(X_{2}, \epsilon_{2}))(h(X_{1}, \eta_{1}) - h(X_{2}, \eta_{2}))} \geq 0$$

    When $g(\cdot, \cdot)$ is non-decreasing and $h(\cdot, \cdot)$ is non-increasing in their first arguments, then it follows that $-h(\cdot, \cdot)$ is non-decreasing and we can write,
    $$\Cov{g(X, \epsilon), h(X, \eta)} = -\Cov{g(X, \epsilon), -h(X, \eta)} \leq 0 .$$
\end{proof}

\begin{proof}[Proof of Proposition~\ref{prop:wc}]
    \hfill\newline
    We show that $\E{f_{\perp} \mid X} = 0$ implies $\theta_{IV}$ is weakly causal. Assume that $\E{f_{\perp} \mid X} = 0$. Using Definition~\ref{def:IV} and the law of total covariance, we can write,
    \begin{equation*}\begin{aligned}\theta_{IV} &= \frac{\E{Yf_{\perp}}}{\E{Df_{\perp}}} = \frac{\Cov{Y, f_{\perp}}}{\Cov{D, f_{\perp}}}\\
    &= \frac{\E{\Cov{Y, f_{\perp} \mid X}} + \Cov{\E{Y \mid X}, \E{f_{\perp} \mid X}}}{\E{\Cov{D, f_{\perp} \mid X}} + \Cov{\E{D \mid X}, \E{f_{\perp} \mid X}}}\\
    &= \frac{\E{\Cov{Y, f(X, Z) \mid X}}}{\E{\Cov{D, f(X, Z) \mid X}}} .\\
    \end{aligned}\end{equation*}

    Now, under Assumption~\ref{assum:exclude} (excludability), if $D(x_{i}, z_{i}, \nu_{i})$ weakly increases in $f(x_{i}, z_{i})$, it follows that $f(x_{i}, z_{i})$ weakly increases in $D$. Therefore, we can express $f(x_{i}, z_{i})$ as a non-decreasing function in $D$, with an additional component that is independent of $\epsilon$ when conditioned on $X$ under Assumption~\ref{assum:exclude} (excludability). We can do the same when $D(x_{i}, z_{i}, \nu_{i})$ weakly decreases in $f(x_{i}, z_{i})$. Then, consider two cases,
    \begin{enumerate}
        \item $Y(d, x, \epsilon)$ is non-decreasing in $d$. Then, under Assumption~\ref{assum:instr_mono} (instrument monotonicity), if $D(x_{i}, z_{i}, \nu_{i})$ is weakly-increasing in $f(x_{i}, z_{i})$, then Lemma~\ref{lemma:cov} (FKG) implies that both $\Cov{Y, f(X, Z) \mid X} \geq 0$ and $\Cov{D, f(X, Z) \mid X} \geq 0$, so $\theta_{IV} \geq 0$. If $D(x_{i}, z_{i}, \nu_{i})$ is weakly-decreasing in $f(x_{i}, z_{i})$, then both $\Cov{Y, f(X, Z) \mid X} \leq 0$ and $\Cov{D, f(X, Z) \mid X} \leq 0$, so $\theta_{IV} \geq 0$

        \item $Y(d, x, \epsilon)$ is non-increasing in $d$. Then, under Assumption~\ref{assum:instr_mono} (instrument monotonicity), if $D(x_{i}, z_{i}, \nu_{i})$ is weakly-increasing in $f(x_{i}, z_{i})$, then Lemma~\ref{lemma:cov} (FKG) implies that $\Cov{Y, f(X, Z) \mid X} \leq 0$ and $\Cov{D, f(X, Z) \mid X} \geq 0$, so $\theta_{IV} \leq 0$. If $D(x_{i}, z_{i}, \nu_{i})$ is weakly-decreasing in $f(x_{i}, z_{i})$, then $\Cov{Y, f(X, Z) \mid X} \geq 0$ and $\Cov{D, f(X, Z) \mid X} \leq 0$, so $\theta_{IV} \leq 0$
    \end{enumerate}

    This shows that $\theta_{IV}$ satisfies Definition~\ref{def:wc} (weakly causal).\\

    The proof given in Proposition 4 of \cite{blandhol_when_2022} shows that if $\E{f_{\perp} \mid X} \not= 0$, then weak causality is not satisfied.
\end{proof}

\begin{proof}[Proof of Corollary~\ref{cor:interpret}]
    \hfill\newline
    Recall from Definition~\ref{def:IV},
    $$\theta_{IV} = \frac{\E{Yf_{\perp}}}{\E{Df_{\perp}}} = \frac{\Cov{Y, f_{\perp}}}{\Cov{D, f_{\perp}}}. $$
    Suppose that $\tilde{f}(X)$ is polynomial in $X$, or linear in the basis $\Phi(X)$. Lemma~\ref{lemma:variance} then implies that $\Cov{D, f_{\perp}} = 0$, so $\theta_{IV}$ is unidentified.

    Therefore suppose that $\tilde{f}(X)$ is nonlinear in $X$. Then,
    \begin{equation*}\begin{aligned}
    \E{\tilde{f}(X) \mid X = x} &= \tilde{f}(X)\\
    &\not= X\E{X^{\top}X}^{-1}\E{X^{\top}\tilde{f}(X)} .
    \end{aligned}\end{equation*}
    So Proposition~\ref{prop:wc} implies that $\theta_{IV}$ is not weakly causal. 
\end{proof}
\end{spacing}

\end{document}